\documentclass[12pt,reqno]{amsart}

\usepackage{enumerate}
\usepackage{graphicx}
\usepackage{epsfig}
\usepackage{xcolor}
\usepackage{amssymb,amsmath,amsthm,amsfonts}

\usepackage[letterpaper, margin=1in]{geometry} 

\usepackage {latexsym}

\usepackage{bbm}
 
\allowdisplaybreaks

\newcommand{\alphaS}{{\alpha^{}_{\mbox{\tiny{\textrm{S}}}}}}

\newcommand{\nc}{\newcommand}

\nc{\sV}{\mathbf{s}}

\nc{\AV}{\mathbf{A}}
\nc{\BV}{\mathbf{B}}
\nc{\DV}{\mathbf{D}}
\nc{\EV}{\mathbf{E}}

\nc{\nullV}{\boldsymbol{0}}

\nc{\colb}[1]{\textcolor{blue}{#1}}

\nc{\ellV}{\vec{\ell}}

\nc{\Cset}{\mathbb{C}}
\nc{\Nset}{\mathbb{N}}
\nc{\Rset}{\mathbb{R}}
\nc{\RR}{\mathbb{R}}
\nc{\Sset}{\mathbb{S}}
\nc{\Zset}{\mathbb{Z}}

\nc{\mEL}{m_{\mathrm{e}}}
\nc{\mPR}{m_{\mathrm{p}}}
\nc{\mZ}{m_{\mathrm{z}}}

\nc{\far}{{\mathfrak{a}}}
\nc{\fm}{{\mathfrak{m}}}

\nc{\les}{\lesssim}
\nc{\nit}{\noindent}
\nc{\nn}{\nonumber}
\nc{\D}{\partial}
\nc{\diff}[2]{\frac{d #1}{d #2}}
\nc{\diffn}[3]{\frac{d^{#3} #1}{d {#2}^{#3}}}
\nc{\pdiff}[2]{\frac{\partial #1}{\partial #2}}
\nc{\pdiffn}[3]{\frac{\partial^{#3} #1}{\partial{#2}^{#3}}}
\nc{\abs}[1] {\lvert #1 \rvert}
\nc{\cAc}{{\cal A}_c}
\nc{\cE}{{\mathcal E}}
\nc{\cF}{{\cal F}}
\nc{\cP}{{\cal P}}
\nc{\cV}{{\cal V}}
\nc{\cQ}{{\cal Q}}
\nc{\cGin}{{\cal G}_{\rm in}}
\nc{\cGout}{{\cal G}_{\rm out}}
\nc{\cO}{{\cal O}}
\nc{\Lav}{{\cal L}_{\rm av}}
\nc{\cL}{{\cal L}}
\nc{\cB}{{\cal B}}
\nc{\cZ}{{\cal Z}}
\nc{\cR}{{\cal R}}
\nc{\cT}{{\cal T}}
\nc{\cY}{{\cal Y}}
\nc{\cX}{{\cal X}}
\nc{\cXT}{{{\cal X}(T)}}
\nc{\cBT}{{{\cal B}(T)}}
\nc{\vD}{{\vec \mathcal{D}}}
\nc{\efield}{\mathcal{E}}
\nc{\mE}{\mathcal{E}}
\nc{\vE}{{\vec \efield}}
\nc{\vB}{{\vec \mathcal{B}}}
\nc{\vH}{{\vec \mathcal{H}}}
\nc{\F}{  \mathcal{F} }
\nc{\ty}{{\tilde y}}
\nc{\tu}{{\tilde u}}
\nc{\tV}{{\tilde V}}
\nc{\Pc}{{\bf P_c}}
\nc{\bx}{{\bf x}}
\nc{\bX}{{\bf X}}
\nc{\bXYZ}{{\bf XYZ}}
\nc{\bY}{{\bf Y}}
\nc{\bF}{{\bf F}}
\nc{\bS}{{\bf S}}
\nc{\dV}{{\delta V}}
\nc{\dE}{{\delta E}}
\nc{\TT}{{\Theta}}
\nc{\dPsi}{{\delta\Psi}}
\nc{\order}{{\cal O}}
\nc{\Rout}{R_{\rm out}}
\nc{\eplus}{e_+}
\nc{\eminus}{e_-}
\nc{\epm}{e_\pm}

\nc{\vnabla}{{\vec\nabla}}
\nc{\G}{\Gamma}
\nc{\w}{\omega}
\nc{\mh}{h}
\nc{\mg}{g}
\nc{\vphi}{\varphi}
\nc{\tlambda}{\tilde\lambda}
\nc{\be}{\begin{equation}}
\nc{\ee}{\end{equation}}
\nc{\ba}{\begin{eqnarray}}
\nc{\ea}{\end{eqnarray}}

\nc{\g}{\gamma}
\nc{\ol}{\overline}

\newtheorem{thm}{Theorem}[section]
\newtheorem{lemm}[thm]{Lemma}
\newtheorem{prop}[thm]{Proposition}
\newtheorem{cor}[thm]{Corollary}

\newtheorem{rmk}[thm]{Remark}

\nc{\pr}{\partial_r}
\nc{\pt}{\partial_t}
\nc{\pT}{\partial_T}
\nc{\pz}{\partial_z}

\nc{\la}{\langle}
\nc{\ra}{\rangle}
\nc{\infint}{\int_{-\infty}^{\infty}}
\nc{\halfwidth}{6.5cm}
\nc{\figwidth}{10cm}
\newcommand{\f}{\frac}

\nc{\nlayers}{L} \nc{\nsectors}{M}
\nc{\indicator}{\mathbf{1}}
\nc{\Rhole}{R_{\rm hole}}
\nc{\Rring}{R_{\rm ring}}
\nc{\neff}{n_{\rm eff}}
\nc{\Frem}{F_{\rm rem}}
\nc{\R}{\mathbb R}
\nc{\C}{\mathbb C}
\nc{\Z}{\mathbb Z}
\nc{\DD}{\Delta}
\nc{\cD}{\mathcal D}
\nc{\lnorm}{\left\|}
\nc{\rnorm}{\right\|}
\nc{\rnormp}{\right\|_{\ell^{p,\eps}}}
\nc{\rar}{\rightarrow}
\nc{\mR}{\mathcal R}
\nc{\oo}{\"o}   

\nc{\os}{\overset{o}}

\nc{\eps}{\epsilon}
\nc{\veps}{\varepsilon}

\nc{\dd}{{\mathrm{d}}}

\nc{\mNULL}{{C_\alpha^{}}}
\nc{\phiNULL}{{C_\beta^{\prime}}}

\nc{\beq}{\begin{equation}}
\nc{\eeq}{\end{equation}}

\date{Version of Jan. 18, 2021; printed: \today \\ 
\copyright (2020) The authors. Reproduction of this preprint is permitted for noncommercial purposes.}

\begin{document}

\title[Dirac operator on Reissner--Weyl--Nordstr\"om black hole spacetimes]
{On the Dirac operator for a test electron in a Reissner--Weyl--Nordstr\"om black hole spacetime}

\author[Kiessling, Tahvildar-Zadeh, Toprak]{Michael K.-H. Kiessling, A. Shadi Tahvildar-Zadeh, Ebru Toprak} 

 \address{Department of Mathematics \\
Rutgers University \\
Piscataway, NJ 08854, U.S.A.}
\email{miki@math.rutgers.edu}
 \address{Department of Mathematics \\
Rutgers University \\
Piscataway, NJ 08854, U.S.A.}
\email{shadit@math.rutgers.edu}
\address{Department of Mathematics \\
Rutgers University \\
Piscataway, NJ 08854, U.S.A.}
\email{et400@math.rutgers.edu}

\begin{abstract}
 The present paper studies the Dirac Hamiltonian of a test electron with a domain of bi-spinor wave functions supported on the static region
inside the Cauchy horizon of the subextremal RWN black hole spacetime, respectively inside the event horizon of the extremal RWN
black hole spacetime.
 It is found that this Dirac Hamiltonian is not essentially self-adjoint, yet has infinitely many self-adjoint extensions.
 Including a sufficiently large anomalous magnetic moment interaction in the Dirac Hamiltonian restores 
essential self-adjointness; the empirical value of the electron's anomalous magnetic moment is large enough.
 In the subextremal case the spectrum of the self-adjoint Dirac operator with anomalous magnetic moment is purely absolutely continuous and 
consists of the whole real line; in particular, there are no eigenvalues.
 The same is true for the spectrum of any self-adjoint extension of the Dirac operator without anomalous magnetic moment interaction,
in the subextremal black hole context.
 In the extremal black hole sector the point spectrum, if non-empty, consists of a single eigenvalue, which is identified.
\end{abstract}

\maketitle
\newpage
\section{Introduction}

 It is well-known that if $M>0$ is the ADM mass of the general-relativistic Reissner--Weyl--Nordstr\"om (RWN) 
spacetime and $Q\neq 0$ its charge, and if $G$ denotes Newton's constant of universal gravitation, then the RWN spacetime features a naked 
singularity when $GM^2 < Q^2$ and a black hole when $GM^2 \geq Q^2$; in the borderline case $GM^2 = Q^2$ one speaks of the {\it extremal} RWN 
black hole, while $GM^2 > Q^2$ is called the {\it subextremal} black hole parameter sector, cf. \cite{HE}.
 For the original publications, see \cite{Reissner}, \cite{WeylARTaxisym}, and \cite{Nordstroem}.

 In their paper ``The general-relativistic hydrogen atom'' \cite{CP} Cohen and Powers rigorously
studied the general-relativistic Dirac operator $H$ (\cite{WeylELECTRONandART}, \cite{Erwin}, \cite{BC})
for a test electron in the RWN spacetime of a point nucleus for both the naked singularity sector and the subextremal black hole sector. 
 They made the startling discovery that in the naked singularity sector $H$ is not well-defined,\footnote{One
   may be tempted to consider this result as a vindication for the widespread opinion that
 ``naked singularities are considered unphysical'' (cf. \cite{GreinerETal}, p.562).
   However, this opinion propagates an unfortunate myth.
     It is based on a misunderstanding of Penrose's weak cosmic censorship hypothesis, which surmises that gravitational collapse
   of cosmic matter does not form a naked singularity. 
   In its strict sense the surmise is wrong, as shown first by Christodoulou \cite{DCviol}, \cite{DCnakedEX} 
for spherically symmetric collapse of scalar matter,
   and most recently by Rodnianski and Shlapentokh-Rothman \cite{RodShlap} for collapsing gravitational waves without symmetry assumption; yet it is
  expected that these scenarios are not generic (this was confirmed for the spherically symmetric scalar case, also by Christodoulou \cite{DCinstab}),
   and that generically (or: typically) a gravitational collapse of cosmic matter will not form a naked singularity.
   However, the point nuclei used in quantum-mechanical models of hydrogenic ions of the kind created in our laboratories 
 are not assumed to have formed through gravitational collapse of charged matter in cosmic proportions.
   In short, the weak cosmic censorship hypothesis, even if generically true, is entirely irrelevant to the problem of general-relativistic 
   hydrogenic ions.}
while for the black hole sector there is a well-defined $H$ but  its essential spectrum is the whole real line, void of any eigenvalues.

 The truly startling part of the discoveries of Cohen and Powers \cite{CP} concerns the naked singularity sector, for it means that 
switching on relativistic gravity destroys the well-defined (i.e. essentially self-adjoint) special-relativistic purely electrical 
hydrogenic ion problem for all parameter values which correspond to empirically known long-lived nuclei ($1\leq Z\leq 118$ and $\mPR\leq M<300\mPR$;
here $\mPR$ denotes the proton mass).
 In more technical language, general-relativistic gravity is not at all a weak perturbation (see \cite{Kato}) of special-relativistic electricity
in the atomic realm, notwithstanding the physics folklore that ``the gravitational interaction between an electron and a nucleus is too weak 
to be significant;'' cf. \cite{EtoS}.

 Also the special-relativistic Dirac Hamiltonian for hydrogenic ions with purely electrical Coulomb interactions is not
always essentially self-adjoint [on the minimal domain $C^\infty_c(\Rset^3\backslash\{0\})^4$].
 We recall that when $Z\in\Nset$ counts the number of elementary charges in the nucleus, then the Dirac Hamiltonian 
is essentially self-adjoint for $Z\leq 118$ \cite{Narnhofer,Thaller}.
 For $119\leq Z\leq 137$ it has a \emph{distinguished} self-adjoint extension,\footnote{The distinguished self-adjoint extention
   is defined by allowing $Z\in\Cset$ and demanding analyticity in $Z$. 
  The real threshold values then become $Z =\sqrt{3}/2\alphaS$ instead of $Z=118$, and $Z=1/\alphaS$ instead of $Z=137$.
 Here, $\alphaS := {e^2}/{\hbar c} \approx 1/137.036$ is Sommerfeld's fine structure constant.\vspace{-15pt}} 
yet for $Z>137$ nobody seems to know which one of uncountably many self-adjoint extensions is physically distinguished.
 The heuristic explanation for the breakdown of analytical self-adjointness in the special-relativistic purely electrical 
hydrogenic ion problem is that the electrical Coulomb attraction between nucleus and electron becomes too strong for the angular momentum
barrier to stabilize, and a collapse of the ground state ensues. 
 Since gravity is generally attractive, one would have expected a worsening of the self-adjointness properties in the general-relativistic
problem, but a complete wipeout was presumably not expected by anyone!

 The problem with the lack of essential self-adjointness of the special-relativistic hydrogenic Dirac Hamiltonian 
goes away, however, if one takes the anomalous magnetic moment $\mu_a$ of the electron into account.
 Indeed, as shown in \cite{Beh,GST} for the special-relativistic hydrogenic problem, adding an anomalous magnetic moment operator to 
the Dirac Hamiltonian of a test electron with purely electrostatic interactions produces an essentially self-adjoint Hamiltonian for 
the electron of any hydrogenic ion, independently of the strength of its non-vanishing anomalous magnetic moment; see
\cite{Thaller,ThallerCHEM} for numerically computed eigenvalues as functions of $Z$ beyond $Z=137$. 
 More recently Belgiorno, Martellini, and Baldicchi \cite{BMB} showed that the Dirac operator for a test electron with anomalous magnetic moment 
is essentially self-adjoint in the naked RWN geometry (only) if $|\mu_a| \geq \frac{ 3}{2} \frac{ \sqrt{G} \hbar}{c}$.
 The empirical $|\mu_a|\approx \mu_{\text{class}} := \frac{1}{4\pi}\frac{e^3}{\mEL c^2}$, which we call
\emph{the classical magnetic moment of the electron}.
 Since \vspace{-5pt}
\be 
 { \sqrt{G} \hbar}/{c} \approx 1.3 \cdot 10^{-18} \mu_{\text{class}} ,
\ee
the hurdle for essential self-adjointness, $|\mu_a| \geq \frac{ 3}{2} \frac{ \sqrt{G} \hbar}{c}$, is easily cleared with the empirical electron
data.
 Here, $\mEL$ is the empirical mass of the electron and $-e$ its charge, 
$c$ is the speed of light in vacuum, and $\hbar$ is Planck's constant divided by $2\pi$, as usual.

 With the Dirac operator for an electron in the naked singularity sector of the RWN spacetime basically understood, 
in this paper we will revisit the problem of the Dirac operator for an electron in the black hole sector of the RWN spacetime family;
we will also include some comparative remarks concerning electrons in the naked singularity sector, though.
 
 Our point of departure is the fact that Cohen and Powers \cite{CP} considered a Dirac Hamiltonian with the minimal domain of 
$C^\infty$ bi-spinor functions with compact support outside the event horizon of a subextremal black hole. 
 They proved that this $H$ is essentially self-adjoint, and that it has the whole real line as its essential spectrum. 
 Thus its discrete spectrum is empty and any eigenvalues would have to be embedded in the continuum.
 Yet in \cite{CP} the absence of eigenvalues is shown altogether. 
 Their result means that a test electron outside the event horizon of a subextremal RWN black hole cannot be in a stationary
bound state. 
 In concert with a result of Weidmann \cite{WeiB} this now implies that the essential spectrum is \emph{purely absolutely continuous},
and so in fact is the spectrum of this Dirac Hamiltonian.

 Upon reflection, it is not too surprising not to find bound states of an electron whose wave function is supported outside the event 
horizon of an RWN black hole.  
 After all, one expects the electron to be swallowed by the black hole unless it escapes to spatial $\infty$.
 The capture of the electron by the black hole is not seen in the treatment by Cohen and Powers, who worked 
with a coordinate system that near the end of the first quarter of the 20th century gave rise to the ``frozen star'' scenario.
  The purpose of this coordinate system was to describe
the collapsing evolution of gravitating masses as seen from spatial infinity, and therefore failed to capture the formation of a black hole. 
 Thus, conceivably, a Dirac bound state in the black hole sector of RWN may exist after all, but it would require 
the domain of the Dirac Hamiltonian to not be restricted to bi-spinor wave functions supported outside the event horizon. 
 Of course, it is often argued on positivistic grounds that physics is not concerned with what goes on inside an event horizon, 
but positivism is merely a form of philosphy which should not be confused with the foundations of physics. 
 Also Werner Israel and his collaborators have long advocated \cite{DIM} investigating what's going on inside 
an event horizon according to general relativity theory. 
 Finster, Smoller, and Yau \cite{FSY} in particular have inquired into time-periodic Dirac bi-spinor wave functions that are supported
both outside and inside the event horizon of a RWN black hole spacetime, and found no nontrivial ones in $L^2$. 
 However, since the region between the Cauchy and the event horizon of a RWN black hole spacetime is not static, insisting on time-periodic
bi-spinors also there seems like asking for too much. 
 In this vein, in this paper we will investigate the Dirac Hamiltonian for a test electron in the  RWN black hole
spacetime with the bi-spinor wave function supported entirely on the \emph{static part} of the region inside the event horizon of the
black hole spacetime, which is a static spherically symmetric spacetime with a naked singularity in its own right --- it is not 
asymptotically flat, though.
 \newpage

 Our results, stated informally, are:\vspace{-3pt}

{\bf Theorem 1}: \emph{The Dirac Hamiltonian $H$ for a test electron in the static interior of a (sub-)extremal 
RWN black hole, if it interacts with the singularity only electrically and gravitationally, is not essentially self-adjoint,
yet has infinitely many self-adjoint extensions.
 In the subextremal case, each self-adjoint extension has a purely absolutely continuous spectrum
that extends over the whole real line.}\vspace{-3pt}
  
{\bf Theorem 2}: \emph{The Dirac Hamiltonian $H$ for a test electron in the static interior of a (sub-) extremal
 RWN black hole, if it interacts with the singularity electrically, gravitationally, and through its anomalous magnetic moment, 
is essentially self-adjoint if and only if $|\mu_a| \geq \frac{ 3}{2} \frac{ \sqrt{G} \hbar}{c}$.
 In the subextremal essentially self-adjoint situation, the unique self-adjoint extension has purely absolutely continuous spectrum 
that covers the whole real line.}\vspace{-3pt}

 Thus, the singularity of the RWN spacetime causes a lack of essential self-adjointness (e.s.a.) if the electron is not shielded from it by 
the event horizon and is assumed to interact only electrically and gravitationally with the singularity, but
e.s.a. is restored if a sufficiently large anomalous magnetic moment of the electron is taken into account. 
 The empirical anomalous magnetic moment of the electron is about $10^{18}$ times larger than the critical value, the same 
{critical value} as found in \cite{BMB} for the naked singularity sector.
 
 However, while in Appendix C of \cite{BMB} it is shown that the general-relativistic hydrogenic Dirac Hamiltonian of a test electron with 
anomalous magnetic moment in the naked singularity sector of the RWN spacetime of a nucleus has infinitely many discrete eigenvalues in the gap 
$(-\mEL c^2, \mEL c^2)$ of its essential spectrum, the essentially self-adjoint operator of an electron with anomalous magnetic moment
inside the Cauchy horizon of a subextremal black hole  has {\it no} eigenvalues at all.
 We will also show that in the extremal case there can be at most one eigenvalue, possibly infinitely degenerate, which we identify.

 In the remainder of this paper we make all this precise.

 In section \ref{sec:RWN} we explain that normal nuclei are associated with the naked singularity sector of the RWN spacetime, 
while hypothetical hyper-heavy nuclei, defined by the inequality ${GM \mEL}>{Z e^2}$ (which means the gravitational 
attraction between a positron and a nucleus overcomes their electrical repulsion),
have to be associated with the RWN black hole sector. 
 We also stipulate our dimensionless notation for both the spacetime and the Dirac operators.

 Section \ref{sec:Dirac} is the main technical section.
 We define the Dirac operators, state our theorems precisely, then present their proofs, using 
strategies of \cite{WeiC}, \cite{KSWW}, and \cite{CP}.
 Some of our proofs are overall very similar to proofs in \cite{KTT} for naked-singularity spacetimes, 
yet details vary.

 We conclude in Section \ref{sec:summ}  and emphasize open problems.

\vspace{-10pt}

\section{The Reissner--Weyl--Nordstr\"om spacetime of a point nucleus}\label{sec:RWN}

 In order to facilitate the comparison of our results with those for hydrogenic ions, including the speculative hyper-heavy ones, 
from now on we think of the central timelike singularity of the RWN spacetime as a proxy for the worldline of a point nucleus at rest. 
 Thus for the charge parameter $Q$ of the RWN spacetime we set $Q=Ze$, where $e>0$ denotes the elementary charge (in Gaussian units),
and where $Z\in\Nset$ counts the number of elementary charges carried by the nucleus.
 We let the ADM mass of the RWN spacetime be the nuclear mass, $M_{\mbox{\tiny{ADM}}}=M = A(Z,N)\mPR$, where $\mPR$ is the proton mass,
$N\in\{0,1,2,...\}$ is the number of neutrons in the nucleus, and $A(Z,N)\geq 1$ the nuclear mass number; moreover, 
$A(Z,N)\approx Z+N$ to within $1\%$ accuracy. 

 All empirically known long-lived atomic nuclei are far away from the black hole regime $GM^2 \geq Z^2e^2$.
 This conclusion extends to hypothetical nuclei with arbitrary large $Z$ if they obey the bounds
$Z\leq A(Z,N) \leq 3Z$ known empirically to hold for all long-lived nuclei with $Z\leq 92$ in the current \emph{chart of the nuclides}.
 Assuming these empirical bounds, essentially $N\leq 2Z$ to within $1\%$ accuracy,
one finds $\frac{GM^2}{Z^2e^2}< 9 \frac{G\mPR^2}{e^2}$, and since $\frac{G\mPR^2}{e^2}\approx 8\cdot 10^{-37} \ll 1$ by many powers of 10, 
also $\frac{GM^2}{Z^2e^2}\ll 1$ and $\frac{GM\mPR}{Ze^2}\ll 1$, thus $\frac{GM\mEL}{Ze^2}\ll 1$.

 On the other hand, hypothetical hyper-heavy nuclei, for which $GM\mEL >  Ze^2$, are associated with the black hole sector of the RWN spacetime.
 For suppose not. 
 Then both $GM\mEL > Ze^2$ and $GM^2 < Z^2e^2$ (the latter condition means we are in the RWN naked singularity regime).
 Since $M=A(N,Z)\mPR$ for all nuclei, we have $M\geq Z\mPR$, and since $\mPR = 1836 \mEL$,
we find that $GM^2 < Z^2e^2$ implies that $1836 A GM\mEL < Z^2e^2$, while $Ze^2< GM\mEL $ implies $1836 A Ze^2 < 1836 A GM\mEL$.
 And so, by transitivity, we have  $1836 A Ze^2 < Z^2e^2$, hence $1836 A < Z$, which is impossible with the empirical $A(Z,N)\approx Z+N$. 

 The previous two paragraphs in concert 
imply that the assumption $N\leq 2Z$ cannot be imposed as a condition when inquiring into hyper-heavy hydrogenic ions.
 Fortunately, neutron stars are in a fair sense examples of gravitationally bound nuclei with $M \approx (Z+N)\mPR$ and $Z$ 
very small while $N$ is very large. 
 Of course, neutron stars are not point-like, yet they are only mentioned as an example of gigantic nuclei in nature not obeying the 
$N\leq 2Z$ rule.
 Hyper-heavy nuclei would not only not obey the $N\leq 2Z$ rule, they would have to be associated with the black hole sector and
therefore de-facto be point singularities covered by an event horizon --- as per Einstein's general relativity theory.
 
\subsection{The naked singularity regime}\label{sec:RWNnaked}
 The electrostatic Reissner--Weyl--Nordstr\"om (RWN) spacetime of a naked point nucleus is spherically symmetric, static, asymptotically flat 
and topologically identical to $\Rset^{1,3}\setminus$ a timelike line, equivalently $\Rset\times (\Rset^3\setminus\{0\})$,
covered by a single global chart of spherical coordinates $(t,r,\vartheta,\varphi)\in \Rset\times\Rset_+\times[0,\pi]\times[0,2\pi)$.
 Here, $r$ is the so-called area radius: every point in the RWN spacetime is an element of a unique orbit of the Killing vector flow 
corresponding to its $SO(3)$ symmetry, and this orbit is a scaled copy of $\Sset^2$ with area $A =: 4\pi r^2$, \emph{defining} $r>0$. 
 Moreover, the variables $\vartheta$ and $\varphi$ are the usual polar and azimuthal angles on $\Sset^2$. 
 In dimensionless units where $r$ is measured in multiples of the electron's reduced Compton wave length $\hbar/\mEL c$, and
$t$ in multiples of $\hbar/\mEL c^2$, its metric has the line element
\begin{equation}\label{dsSQR}
ds^2 = - f^2(r) \dd t^2+ f^{-2}(r) \dd r^2 + r^2 \dd\Omega^2, \qquad  
\end{equation}
where $\dd\Omega^2 = \dd\vartheta^2 +\sin^2\vartheta\dd\varphi^2$ is the line element on $\Sset^2$, and where
\begin{equation}\label{fOFr}
f^2(r) = 1 - 2\frac{GM\mEL}{\hbar c}\frac{1}{r}
+ \frac{G\mEL^2}{\hbar c}\frac{Z^2e^2}{\hbar c}\frac{1}{r^2}.
\end{equation}
  Here, $M= A(Z,N)\mPR$, and we note that
$\frac{G\mEL^2}{\hbar c}\approx 1.79\cdot 10^{-45}$ and $\frac{G\mPR\mEL}{\hbar c}\approx 3.3\cdot 10^{-42}$; incidentally,
$\frac{G\mPR^2}{\hbar c}\approx 6.0\cdot 10^{-39}$.
 Also, $\frac{e^2}{\hbar c} \approx 1/137.036$ is Sommerfeld's fine structure constant.

 The known long-lived nuclei, for which $A(Z,N)\approx Z+N$ and $Z\leq A(Z,N) {\leq} 3Z$, are 
associated with the naked singularity sector of the RWN spacetimes, i.e. $f^2(r)>0\ \forall r>0$.

\subsection{The black hole regime}\label{sec:RWNbh}
 The RWN spacetime features a black hole if there is at least one value of $r>0$ for which $f^2(r) =0$. 
  Since $f^2(r)$ is a quadratic polynomial in $1/r$, viz. $f^2(r)= \frac{1}{r^2} (r-r_+^{})(r-r_-^{})$, with the zeros formally given by 
\begin{equation} \label{rpm}
 r_\pm = \tfrac{GM\mEL}{\hbar c} \left( 1\pm \sqrt{1- \tfrac{Z^2e^2}{GM^2}}\right) ,
 \end{equation}
and those are real {if and only if} $\tfrac{Z^2e^2}{GM^2}\leq 1$. 
 If $\tfrac{Z^2e^2}{GM^2} = 1$, one says the asymptotically flat spacetime contains
an \emph{extremal} black hole; if $\tfrac{Z^2e^2}{GM^2} < 1$, the asymptotically flat spacetime contains a \emph{subextremal} black hole.
In the extremal case, $r_+^{}=r_-^{} = \frac{GM\mEL}{\hbar c} =:r_0^{}$, and then
\be \label{r0}
f^2(r) = \left(1-\frac{GM\mEL}{\hbar c}\frac1r\right)^2 = \frac{1}{r^2}\left(r-r_0^{}\right)^2.
\ee

 Continuing an asymptotically flat RWN black-hole spacetime analytically, one finds two static regions: 
either $r> r_+^{}$ or $r<r_-^{}$; this is true even for the extremal case when $r_+^{} = r_-^{}$.
 The maximal analytically extended spacetime even has infinitely many copies of such regions.
 We will be concerned with spacetimes given by one copy of the inner static region.

\section{The Dirac operators}\label{sec:Dirac}\vspace{-10pt}

 In this section we formulate the Dirac operator for a test electron with or without anomalous magnetic moment in the RWN spacetime of a naked 
point nucleus.
 For the sake of definiteness, we will define the electrons' anomalous magnetic moment 
as identical to its highly accurate approximation $\mu_a = -\mu_{\text{class}}$. 
 However, we multiply $\mu_a$ by an amplitude $\far$: for $\far=0$ we obtain the Dirac operator for a point electron 
without anomalous magnetic moment, whereas $\far=1$ if the electron's anomalous magnetic moment is taken into account. 
 By varying $\far$ continuously we can inquire into the threshold for essential self-adjointness.

 Electrons with wave function restricted to the region  $r> r_+^{}$ on the subextremal RWN black hole spacetime
were studied in \cite{CP}; no bound states exist, then.
 We will investigate electrons with wave function restricted to the region $r<r_-^{}$ on the subextremal RWN black hole spacetime, 
i.e. to the spacetime given by \eqref{dsSQR}, \eqref{fOFr}, with $(t,r,\vartheta,\varphi)\in \Rset\times (0,r_-^{})\times[0,\pi]\times[0,2\pi)$.
 We will also study wave functions on the extremal RWN black hole spacetime, supported inside the event horizon at $r_0^{}\;(=r_-^{}=r_+^{})$.
 For the purpose of comparison, we will also recall the Dirac operator defined on the naked singularity sector. 

 Due to the spherical symmetry and static character of the spacetimes, the 
Dirac operator $H$ of a test electron in the curved space whose line element $\dd s^2$ is given by (\ref{dsSQR}) 
separates in the spherical coordinates and their default Cartan frame \cite{CP}. 
 More precisely,  $H$ is a direct sum of so-called radial partial-wave Dirac operators
$H^{\mbox{\tiny{rad}}}_{k} := \mEL c^2 K^{\far}_{k}$, $k\in\Zset\backslash\{0\}$, with
\begin{align}\label{Kk} 
K_{k}^{\far}
:= 
\left[\begin{array}{cc} f(r) - Z\alphaS\frac{1}{r} & \left[\frac{k}{r} -Z\alphaS^2\frac{\far}{4\pi} \frac{1}{r^2}\right] f(r) -f^2(r)\tfrac{d}{dr} 
\\
\left[\frac{k}{r}  -Z\alphaS^2\frac{\far}{4\pi} \frac{1}{r^2}\right] f(r) + f^2(r) \tfrac{d}{dr} 
& - f(r) - Z\alphaS\frac{1}{r} 
	\end{array}\right],
\end{align}
which act on two-dimensional bi-spinor wave function subspaces.
 The spectrum of $H$ is the union of the spectra of the $H^{\mbox{\tiny{rad}}}_{k}$. 
 This reduces the problem to studying the spectrum of $K^{\far}_{k}$.\vspace{-10pt}

\subsection{Point nucleus as naked singularity of static spacetime}\label{sec:DiracRWNnaked} \vspace{-5pt}
 In this case the bi-spinor wave functions are supported on $\Rset^3$ minus a point. 
 The radial Hilbert space consists of pairs $g(r):=\big(g_1(r) , g_2(r)\big)^T$ equipped with a weighted $L^2$ norm given~by
\be \label{fnormNAKED}
\| g\|^2 := \int_{0}^{\infty} \frac{1}{f^2(r)}  \Big( |g_1(r)|^2 + |g_2(r)|^2 \Big) \dd r. 
\ee

 As mentioned in the introduction, Cohen and Powers \cite{CP} proved that for $\far=0$ the Dirac Hamiltonian is not 
essentially self-adjoint, but has uncountably many self-adjoint extensions. 
 Belgiorno et al. \cite{BMB} subsequently showed that $H$ is essentially self-adjoint on the domain of $C^\infty$ bi-spinor
wave functions which are compactly supported away from the singularity at $r=0$ whenever $\far$ is large enough, viz. if
$\far \mu_{\text{class}} \geq \frac{3}{2} \frac{ \sqrt{G} \hbar}{c}$.

\subsection{Point nucleus as singularity in static interior of black hole spacetime}\label{sec:DiracRWNbh}

 In this case the bi-spinor wave functions are supported on $\Sset^2\times(0,r_-^{})$.
 The radial Hilbert space consists of pairs $g(r):=\big(g_1(r) , g_2(r)\big)^T$ equipped with a weighted $L^2$ norm given~by
\be \label{fnormBH}
\| g\|^2 := \int_{0}^{r_-^{}} \frac{1}{f^2(r)}  \Big( |g_1(r)|^2 + |g_2(r)|^2 \Big) \dd r. 
\ee
 In the extremal case,
$r_-^{}=r_0^{}$.

 We change variables $r\mapsto x$ such that
\begin{align} \label{eq:cv}
f^2(r) \frac{d}{dr}= \frac{d}{dx},
\end{align}
with $x=0$ when $r=0$, which for the subextremal sector yields
\be\label{xOFrSE}
 x = r + \frac{r_+^2}{r_+^{}-r_-^{}}\ln\left(1-\frac{r}{r_+^{}}\right) -  \frac{r_-^2}{r_+^{}-r_-^{}}\ln\left(1-\frac{r}{r_-^{}}
\right);
\qquad r < r_-^{};
\ee
in the extremal limit $r_-^{}\nearrow {r_0^{}} \ \&\ r_+^{}\searrow {r_0^{}}$ this becomes
\be\label{xOFrEX}
 x = r_0^{}\left[ \frac{1}{1-\frac{r}{r_0^{}}} + 2\ln\left(1-\frac{r}{r_0^{}}\right) - \left(1-\frac{r}{r_0^{}}\right) \right], \qquad r < r_0^{}. 
\ee

 Note that $x\to\infty$ when $r\nearrow r_-^{}$, respectively when $r\nearrow r_0^{}$.
 This maps $K_k^{\far}$ into 
\begin{align} \label{Ktilde}
\widetilde{K}_k^{\far}  &= 
\left[\begin{array}{cc} f(r(x)) - Z\alphaS\frac{1}{r(x)} & \left[\frac{k}{r(x)} -Z\alphaS^2\frac{\far}{4\pi} \frac{1}{r^2(x)}\right] f(r(x)) -\tfrac{d}{dx} 
\\
\left[\frac{k}{r(x)}  -Z\alphaS^2\frac{\far}{4\pi} \frac{1}{r^2(x)}\right] f(r(x)) + \tfrac{d}{dx} 
& - f(r(x)) - Z\alphaS\frac{1}{r(x)} 
	\end{array}\right],\\
	 & =: 
 \left[\begin{array}{cc}  a(x) - b(x)  & kc(x) -\far d(x) - \frac{d}{dx} \\  kc(x) -\far d(x) + \frac{d}{dx}&  -a(x) - b(x) 
	\end{array}\right]  \label{Ktildeab}
\end{align}
with the inner product 
\begin{align}\label{inner}
 \la g,h   \ra =   \int_{0}^{\infty} \Big( g_1(r(x)) \bar{h}_1(r(x)) + g_2(r(x)) \bar{h}_2(r(x)) \Big) dx {.}
\end{align}

\subsubsection{Electron without anomalous magnetic moment}\label{sec:noANOM}

\begin{thm} \label{th:self}
 The operator $\widetilde{K}_k^{0} $ given by (\ref{Ktilde}) with ${\far}=0$ has uncountably
many self-adjoint extensions for both the subextremal and the extremal black-hole sector.  
\end{thm}
\begin{proof}
 We use the strategy of \cite{CP}, \cite{KTT} for the naked singularity spacetimes.
 We start with the subextremal case.
 Note that, with the change of variable \eqref{xOFrSE} one has\footnote{Here, ``$f(x) \sim g(x)$ as $x\to x_\ast$'' 
means $\exists C>0$ such that $f(x)/g(x) \to C$ as $x\to x_\ast$, where $x_\ast =0$~or~$\infty$.}
$(r-r_-^{})\sim e^{-\varkappa x}$ where $\varkappa = {\frac{r_+-r_-}{r_-^2}}$ as 
$ x\rightarrow \infty$ and  $ r \sim x^{1/3}$ as $x \rightarrow 0$. 
 Therefore, in the subextremal case the operator \eqref{Ktildeab} features $a(x) \sim x^{-1/3}$, $b(x) \sim x^{-1/3}$ 
and $c(x) \sim x^{-2/3}$ as $ x \rightarrow 0$. 
 Furthermore, as $x\rightarrow \infty$, we have $f(r(x)) \sim e^{-\varkappa x} $ as well as $ a(x), c(x) \sim e^{-\varkappa x}$,  and
$b(x) \rightarrow \frac{ Z \alpha_s}{r_-} $.

 Let  $\widetilde{K}^{0*}_k$ be the adjoint operator of $\widetilde{K}^0_k$. 
 The domain $\cD(\widetilde{K}^0_k)$ comprises all $C^{\infty}$ functions of compact support in $(0,\infty)$,
 and $\cD(\widetilde{K}^{0*}_k)$ includes the functions $f$ which together with $f^{\prime}$ are integrable in any compact subset of $[0,\infty)$.  
 On $\cD(\widetilde{K}^{0*}_k)$ we now define the sesquilinear form
\be
 [g, h]: = \la \widetilde{K}^{0*}_k g, h \ra - \la g, \widetilde{K}^{0*}_k h \ra ,
\ee
with $\la \cdot, \cdot  \ra$ defined in \eqref{inner}. 
 By Theorem~4.1 in \cite{WeiC}, $ \widetilde{K}^{0*}_k |_{\cD}$ is a self-adjoint extension of $\widetilde{K}^0_k$ iff 
  \begin{enumerate} 
 \item[$i)$]$ \cD(\widetilde{K}^0_k) \subset  \cD \subset \cD( \widetilde{K}^{0*}_k)$
\item[$ii)$]  $[g, h] =0$ for all $g,h \in \cD$ 
 \item[$iii)$] if $g \in \cD( \widetilde{K}^{0*}_k)$ and  $[g, h] =0$ holds for every $h \in \cD$ then $g \in \cD$. 
 \end{enumerate} 

 Now consider the spaces in which $[g,h]=0$. 
 Take $g\in \cD(\widetilde{K}^{0*}_k)$ so that $ \widetilde{K}^{0*}_k g  = \psi $ for some $\psi \in L^2([0, \infty)$. 
 Since $ \cD({\widetilde{K}}_k^0) \subset C_c^{1}([0, \infty))$, $g \in AC([{x_1,x_2}])$ for each $ 0 \leq {x_1<x_2}< \infty$, and 
so we can integrate to obtain
\begin{align}\label{ga}
g_1(x) = e^{{-}\nu(x)} \Big( g_1(0) {+} \int_0^{x} e^{{+}\nu(y)} [(  a(y) {+}  b(y) )g_2(y) + \psi_2(y)] dy \Big), \\
g_2(x) = e^{{+}\nu(x)} \Big( g_2(0) + \int_0^{x} e^{{-}\nu(y)} [(  a(y) {-}  b(y) )g_1(y) {-} \psi_1(y)] dy \Big) 
\label{gb}
\end{align}
for each $ x \leq {x_2} < \infty$, where $ \nu(x) = \int_0^{x} kc(y) dy \sim x^{1/3} \rightarrow 0$ 
as  $x \rightarrow 0$.
 We also need   $b(x), a(x) \in L^2 ([0, {x_2}])$, ${x_2} < \infty$, for $g_1$, $g_2$ to be defined.
 Integration by parts yields
\begin{align}
[g, h] 
& = \lim_{{x_1} \rightarrow 0} \lim_{{x_2} \rightarrow \infty}  \int_{{x_1}}^{{x_2}} {\Big(} ( \widetilde{K}^{0*}_k g)_1 \overline{h_2} - 
(\widetilde{K}^{0*}_k g)_2 \overline{h_1} + g_1 \overline{(\widetilde{K}^{0*}_k h)_2} - g_2 \overline{(\widetilde{K}^{0*}_k h)_1} \Big) 
{dx}
\\
& = \lim_{{x_1} \rightarrow 0} \lim_{{x_2} \rightarrow \infty} \Big[ g_1({x_2}) \overline{h_2({x_2})} - g_2({x_2}) \overline{h_1({x_2})} -  
g_1({x_1}) \overline{h_2({x_1})} + g_2({x_1})\overline{h_1({x_1})} \Big] \\   & = g_2(0) \overline{h_1(0)} - g_1(0) \overline{h_2(0)}.
\end{align}
 To obtain the last equality we used \eqref{ga}, \eqref{gb}, and the fact that $g,h \in L^2([0,\infty))$.

 Thus any symmetric extension requires $g_2(0) \overline{h_1(0)} - g_1(0) \overline{h_2(0)}= 0$. 
 By taking $g=h$ one sees that this is possible iff one of $g_1(0)$ and $g_2(0)$ is a real multiple of the other.
 Therefore, 
\begin{align} \label{Domdef}
{\widetilde{K}}^0_{k;\theta} :=  
\widetilde{K}^{0*}_k| \cD_\theta,\,\, \text{where}  \,\ 
\cD_\theta = \{g\in \cD(\widetilde{K}^{0*}_k): g_1(0) \sin \theta + g_2(0) \cos \theta =0\} ,
\end{align}
gives a symmetric extension for any $0 \leq \theta < \pi$, cf. \cite{CP}. 
 Note that $\cD_\theta$ satisfies both the conditions $i)$ and $ii)$.
 To see that condition $iii)$ is also satisfied, let $h \in \cD_\theta$, then
  \begin{align}
 [g,h] =0 \iff 
g_2(0) \overline{h_1(0)} - g_1(0) \overline{h_2(0)} =0 \iff 
\frac{\overline{h_2(0)}}{\overline{h_1(0)} } = \frac{g_2(0)}{ g_1(0)} = - \tan \theta \in \R.
  \end{align} 
 This completes the proof for the subextremal case (cf. the proof of Thm.3.6 in \cite{KTT}).
 
 For the extremal case, we consider the change of variable \eqref{xOFrEX} and consider the operator \eqref{Ktildeab}. 
Note that the above argument is valid if  $ \nu(x) = \int_0^{x} kc(y) dy \rightarrow 0$ as  $x \rightarrow 0$; 
$b(x), a(x) \in L^2 ([0, {x_2}])$, ${x_2} < \infty$.

 One can easily see that $c(x) \sim x^{-2/3}$ as $ x \rightarrow 0$ also in the extremal case. 
Furthermore, $b(x), a(x) \sim x^{1/3}$ as $x \rightarrow 0$, and both are continuous in $[0, b]$. 
 Hence, the proof applies similarly.    \end{proof}

\begin{rmk}
 We remark that the deficiency indices of ${\widetilde{K}}_{k}^0$ are $(1,1)$, for 
$[g,h]$ is the difference of two positive rank-one bilinear forms.
 This already implies that an orbit of self-adjoint extensions must exist.
 The proof of Thm.3.1 identifies these.
\end{rmk}

 Our next theorem identifies the essential spectrum of any self-adjoint extension of the Dirac operator acting on bi-spinor wave-functions supported 
inside the inner horizon of either the subextremal or the extremal black-hole spacetime. 
\begin{thm}\label{essK} For each $\theta$, one has $\sigma_{ess} ( {\widetilde{K}}_{k;\theta}^0)= \mathbb{\R}$.
\end{thm}
To prepare the proof of Theorem \ref{essK}, as in \cite{KTT} we recall the following lemma from \cite{CP}.
\begin{lemm}\label{cplem} Let 
\be
{D :}= \left[ \begin{array}{cc} 0 & - \frac{d}{dx} \\ \frac{d}{dx} & 0  \end{array} \right] 
\ee
{be defined on the $C^{\infty}$ two-component functions of compact support in the positive real half-line.
 Now take the closure of this operator in $(L^2(\Rset_+))^2$ with the boundary condition $f_1(0) \sin \theta + f_2(0) \cos \theta=0$ 
at $x=0$, {denoted} $D_\theta$.}
 Let $A$ be the operator 
\be
A = \left[ \begin{array}{cc}a_{11}(x) & a_{12}(x) \\  a_{21}(x) &  a_{22}(x)	  \end{array} \right] ,
\ee
where the $a_{ij}$ are functions in $L^2([0,b])$ for all $ 0 <b < \infty$ and $a_{ij}(x) \rightarrow 0$ as $ x \rightarrow \infty$. 
 Then $A$ is $D_{\theta}$ compact.  
\end{lemm}
\begin{proof}[Proof of Theorem~\ref{essK}] 
 We prove the theorem explicitly for the subextremal case.
 Yet note that the extremal case follows verbatim after setting $r_-\to r_0$.

We split the operator ${\widetilde{K}}_{k}^0$ in \eqref{Ktilde} as
\begin{align}
{\widetilde{K}}_{k}^0  & = \left[ \begin{array}{cc} - \frac{Z \alpha_s}{r_-} & kc(x)- \frac{d}{dx}  \\ 
kc(x) + \frac{d}{dx}  & - \frac{Z \alpha_s}{r_-}
	  \end{array} \right] + \left[\begin{array}{cc} a(x) - \Big[b(x) - \frac{Z \alpha_s}{r_-} \Big] &   0 \\
  0 & - a(x) -  \Big[b(x) - \frac{Z \alpha_s}{r_-}\Big]
	\end{array}\right]  \\
&=:  {\widetilde{K}}^{00}_{k} + V .
\end{align}
 Note that Theorem~\ref{th:self} is valid when $a(x)=0$ and $b(x)= \frac{Z \alpha_s}{r_-}$.
 Therefore, ${\widetilde{K}}^{00}_{k}$ has deficiency indices $(1,1)$ and has multiple self-adjoint extensions similar to ${\widetilde{K}}_k^0$.
 We define these self-adjoint extensions as ${\widetilde{K}}^{00}_{k;\theta}$ similar to ${\widetilde{K}}^0_{k;\theta}$, cf. \eqref{Domdef}. 
 We define the following  Weyl sequence for ${\widetilde{K}}^{00}_{k;\theta}$: let $w = - \frac{Z \alpha_s}{r_-} - \lambda$, with
any $\lambda\in\Rset$, then
 \be \label{weyl}
f_{n, \lambda} (x) = \frac{1}{2 n^{\f 32}} x e^{- \frac{x}{2n} + i x w }
\left[{\scriptsize{\begin{array}{c} 1\\ -i 
	\end{array}}}\right];\quad n\in\Nset.
\ee
  We have that $ \| f_{n, \lambda} \|_{(L^2(\Rset_+))^2} =1$, $  f_{n, \lambda} (x) \rightarrow 0 $ weakly, and 
$\| ( {\widetilde{K}}^{00}_{k;\theta} - \lambda) f_{n, \lambda} (x) \|_{(L^2(\Rset_+))^2} \rightarrow 0 $ as $n\to\infty$.
 Hence, any $\lambda\in\Rset$ is in the essential spectrum of ${\widetilde{K}}^{00}_{k;\theta}$, and so $\sigma_{ess} 
({\widetilde{K}}^{00}_{k;\theta}) = \Rset$.

 Next we will show that $V$ is ${\widetilde{K}}^{00}_{k;\theta}$  compact. 
 This is done essentially verbatim to the pertinent part in the proof of Lemma 3.12 in \cite{KTT}.
 We define $ \xi(x) = - \int_0^x kc(y) dy$ for $0 \leq x \leq 1$ and $ \xi(x) = \xi(1)$ for $x>1$.
 Then the following matrix is bounded,
\be\label{S}
S= \left[ \begin{array}{cc} e^{-\xi(x)} &0 \\  0 &  e^{\xi(x)}
	  \end{array} \right]  .
\ee

 Assume that $\|g_{(n)}\|_{(L^2(\Rset_+))^2} , \| {\widetilde{K}}^{00}_{k;\theta} g_{(n)}\|_{(L^2(\Rset_+))^2} $, $n\in\Nset$,
 are bounded sequences.
 Then $\| Sg_{(n)} \|_{(L^2(\Rset_+))^2}$ and  $\| D_\theta S g_{(n)}\|_{(L^2(\Rset_+))^2}$ {are also bounded,
the first one is because $S$ is bounded and the latter one is by the fact that 
$ D_\theta S = S^{-1} S D_\theta S =  S^{-1} ({\widetilde{K}}^{00}_{k;\theta} +W)$ for some bounded $W$.}
 Moreover, one can check that $VS^{-1}$ is $D_\theta$ compact by Lemma~\ref{cplem}. 
 Hence, 
\be 
 VS^{-1} S g_{(n)} = Vg_{(n)}
\ee 
has a convergent subsequence. 
 This proves that $V$ is $D_\theta$, and ${\widetilde{K}}^{00}_{k,\theta}$ compact. 
\end{proof}

 We next show that the essential spectrum has no singular continuous part.
 As in \cite{KTT}, for this we will need the following Theorem from \cite{WeiB,WeiC}.
\begin{thm}\label{th:weid} (Weidmann) Let 
\be
\tau:=  \left[\begin{array}{cc}  0  & - \frac{d}{dx} \\   \frac{d}{dx} &  0
	\end{array}\right] + P_1(x) + P_2(x) 
\ee
be defined on $ (a , \infty)$. 
 Further assume that  $|P_1(x)| \in L^1 (c, \infty)$  for some $ c \in (a , \infty)$, and  $P_2(x)$ is of 
bounded variation in $[c, \infty)$ with 
\be 
\lim_{x \rightarrow \infty} P_2(x) = \left[\begin{array}{cc}  \fm_{+}  & 0 \\   0 &  \fm_{-}
	\end{array}\right] \,\,\,\,\,\ \text{for} \,\,\,\,\,\,  \fm_{-} \leq  \fm_{+} .
\ee
Then every self-adjoint realization $A$ of $\tau$ has purely absolutely continuous spectrum in 
$(-\infty, \fm_{-}) \cup (\fm_{+}, \infty)$. 
\end{thm}

\begin{prop} \label{acK} 
In both the subextremal and the extremal case, 
$\mathbb{R}\setminus\{ -\frac{\alpha_s Z} {r_-}\} \subset \ \sigma_{ac}({\widetilde{K}}_{k;\theta}^0)$.
\end{prop}

\begin{proof} 
Recall that {${\widetilde{K}}^{0}_{k,\theta}$} is in the form of $\tau$, with $P_1(x) =0$, and with
$P_2(x) = P_2^\far(x)$ for $\far=0$, where
\begin{align} \label{def:p2}
{P_2^\far(x) =  \left[\begin{array}{cc} f(r(x)) - Z\alphaS\frac{1}{r(x)} & \left[\frac{k}{r(x)} -Z\alphaS^2\frac{\far}{4\pi} \frac{1}{r^2(x)}\right] f(r(x))
\\
\left[\frac{k}{r(x)}  -Z\alphaS^2\frac{\far}{4\pi} \frac{1}{r^2(x)}\right] f(r(x)) 
& - f(r(x)) - Z\alphaS\frac{1}{r(x)} 
	\end{array}\right];}
\end{align}
here, both $f(r(x))$ and $r(x)$ are continuously differentiable and hence of bounded variation. 
 Furthermore, 
\be
\lim_{x \rightarrow \infty} f(r(x))=0 , 
\,\,\,\,\mbox{and} \,\,\,\, 
  \lim_{x \rightarrow \infty} r(x) = r_*,
\ee 
where $r_*= r_-$ in the subextremal, and $r_*= r_0$ in the extremal case.
 This implies
\be \label{lim0}
\lim_{x \rightarrow \infty} P_2^\far(x) = \left[\begin{array}{cc}  -\frac{\alpha_s Z} {r_*} & 0 \\   0 &  -\frac{\alpha_s Z} {r_*} \end{array}\right],
\ \forall\ \far\geq 0. 
\ee
 Hence, the spectrum of ${\widetilde{K}}^0_{k;\theta}$ is purely absolutely continuous on $\mathbb{R}\setminus \{ -\frac{\alpha_s Z} {r_*} \}$.
\end{proof}
\begin{cor}\label{cor:sc} 
  The singular continuous spectrum $\sigma_{sc}(K_{k;\theta}^{\text{0}})= \emptyset$ in both the subextremal and the extremal case.
\end{cor}
\begin{proof}
 By Proposition~\ref{acK} and Theorem~\ref{essK} the essential spectrum is the closure of $\sigma_{ac}(K_{k,\theta}^0)$. 
 Since the singular continuous spectrum is a subset of the essential spectrum, and since the interior of the
essential spectrum here is purely absolutely continuous, a non-empty $\sigma_{sc}(K_{k;\theta}^{\text{0}})$ 
would have to consist of the single point $-\alphaS Z/r_*$, which is impossible.
\end{proof}

 The results obtained so far show the absence of a discrete spectrum, but not the absence of point spectrum. 
 Obviously, any point spectrum would have to consist of a single eigenvalue, $-\alphaS Z/r_*$, which could be infinitely degenerate.
 We next show that in the subextremal case, $-\alphaS Z/r_*$ is not an eigenvalue.

 \begin{thm}\label{eigen}
   In the subextremal case  $ {\widetilde{K}}_{k;\theta}^0 $ has no eigenvalues.
\end{thm} 
 For the proof of Theorem \ref{eigen} we will utilize the following Lemma of Cohen \&\ Powers \cite{CP}. 
\begin{lemm}\label{noeig} Let $\mathcal{H}$ be a Hilbert space. 
 Let $V_t$ for $t>a$ be a bounded linear operator on $\mathcal{H}$ so that $V_tf$ is 
continuous in $t$ for each $f \in \mathcal{H}$.
 Suppose $f(t)$ solves the differential equation  
\begin{align}
\frac{df(t)}{dt} = V_t f(t) ,\,\,\,\,\,\ \mbox{\rm{for}}\ t>a
\end{align}
where the derivative exists in the strong sense.
 Suppose $ \int_a^{\infty} \| V_t\| dt = C < \infty$.
  Then the $\lim_{t \rightarrow \infty} f(t)$ exists, and if this limit is zero then $f(t) =0 $ for all $t$. 
\end{lemm} 

 \begin{proof}[Proof of Theorem~\ref{eigen}]  
 We recall that by Proposition~~\ref{acK}, $\lambda = - \frac{Z \alpha_s}{r_-}$ is the only possible value which may be an eigenvalue. 
 Hence, it is enough to show that if $(\widetilde{K}_{k;\theta}^0+\frac{Z \alpha_s}{r_*} I ) g =0$ and $g \in L^2$, then $g$ is identically zero.

 Now note that 
\begin{align} \label{eigop}
{\widetilde{K}}_{k;\theta}^0 + \frac{Z \alpha_s}{r_-} I & =
 \left[ \begin{array}{cc} -b(x) + \frac{Z \alpha_s}{r_*} & - \frac{d}{dx}  \\  \frac{d}{dx}  
& -b(x) + \frac{Z \alpha_s}{r_*}    \end{array} \right] + 
 \left[ \begin{array}{cc} a(x)  &  kc(x)  \\  kc(x)
& -a(x)   \end{array} \right].
\end{align}
 Recall that $a(x), c(x) \rightarrow 0$  as $ x \rightarrow \infty $. 
  More specifically, $a(x),c(x)  \sim e^{-\varkappa x}$ in the subextremal case, 
whereas in the extremal case $a(x),c(x)\sim x^{-1} $, as $x \rightarrow \infty$. 
 
 Let $ \eta(x) = \int_A^{x}  (-b(y) + \frac{Z \alpha_s}{r_*}) dy $ for some $A >0$.
 Recall that $b(y)$ is continuous away from zero and therefore, $ \eta^{\prime}$ is defined. 
Next, note that ignoring 
the second matrix at r.h.s.\eqref{eigop}, the so truncated eigenvalue problem is locally solved by 
$g_\pm(x) = A_\pm e^{\pm i \eta(x)} \g_\pm$, where $\g_\pm = [1,\mp i]^T$; yet note that $g_\pm(x)$ is not in $L^2$.

 The full eigenvalue problem $(\widetilde{K}_{k;\theta}^0+\frac{Z \alpha_s}{r_*} I ) g =0$, and for $g \in L^2$, can now be addressed with the 
help of the method of variation of constants. 
 Thus we make the ansatz
\begin{align}\label{eigenf}
g(x) = \left[\begin{array}{c} u(x) e^{ i  \eta(x)} + v(x) e^{-i\eta(x)}\\
 -i u(x) e^{ i \eta(x)} + i v(x) e^{-i\eta(x)}
\end{array}\right]. 
\end{align}
 Inserting (\ref{eigenf}) into  $(\widetilde{K}_{k;\theta}^0 +\frac{Z \alpha_s}{r_*} ) g=0 $ we obtain 
\begin{equation}
\hspace{-5pt}
  \left[\!\!\begin{array}{c} i u^{\prime}(x) e^{ i  \eta(x)} - i v^{\prime} (x) e^{- i  \eta(x)}  \\
 u^{\prime}(x) e^{ i \eta(x)} + v^{\prime} (x) e^{-i\eta(x)}
\end{array}\!\!\right] 
\! =\! \left[\!\!\begin{array}{c} u(x) e^{ i  \eta(x)} [a(x) - i k c(x)] + v(x) e^{-i\eta(x)} [a(x) + i kc(x)]  \\
 u(x) e^{i\eta(x)} [ia(x) + kc(x)] +  v(x) e^{-i\eta(x)} [ - i a(x) + kc(x)] 
\end{array}\!\!\right] \!.
\hspace{-15pt}
\end{equation}
 Written in terms of $\frac{d}{dx} (u,v)^T$ this yields
\begin{align} 
 \frac{d}{dx} \left[\!\! \begin{array}{c} u(x) \\ v(x) \end{array}\!\! \right]
= 
 \left[ \begin{array}{cc} 0 & e^{-2 i \eta(x)}  [ia(x) + kc(x)]   \\  e^{2i \eta(x)} [-ia(x) + kc(x)] & 0   \end{array} \right] 
\left[\!\! \begin{array}{c} u(x) \\ v(x) \end{array}\!\! \right].
\end{align}

Recall that in the subextremal case  $a(x), c(x) \sim e^{-\varkappa x}$ and therefore, by Lemma~\ref{noeig}, we obtain $u = 0 = v$.
 This finishes the proof. 
\end{proof}

\begin{rmk} 
 Lemma~\ref{noeig} can be applied only to the subextremal case, where we have $r - r_- \sim e^{-\varkappa x}$ as $ x \rightarrow \infty$, 
and $e^{-\varkappa x}$ is integrable at $\infty$.
 On the other hand, in the extremal case, $r - r_0 \sim x^{-1}$ as $ x \rightarrow \infty$ and this does not satisfy the integrability condition 
in Lemma~\ref{noeig}.
\end{rmk}

 As an immediate consequence of our Proposition \ref{acK}, and Theorem \ref{eigen}, we have
\begin{cor}\label{cor:AC} 
 In the subextremal case the essential spectrum, given by the whole real line, is purely absolutely continuous:
$\sigma_{ac}(K_{k;\theta}^{\text{0}})= \Rset$.
\end{cor}
%

\subsubsection{Electron with anomalous magnetic moment}\label{sec:ANOM}
\hfill

 We now address the Dirac Hamiltonian for an electron with anomalous magnetic moment in the static interior of a RWN black-hole spacetime. 
\begin{thm} 
 In both the subextremal and the extremal case, the operator ${\widetilde{K}}_k^{\far}$ given by (\ref{Ktilde}) is essentially self-adjoint if 
 $\far\frac{e^3}{4\pi  \mEL c^2}  \geq \frac32 {\sqrt{G}\hbar}/{c}$.
\end{thm}
\begin{proof} 
 We will show that the limit point case (LPC) is verified both, in the right neighborhood of $x = 0$, and in the left neighborhood of $x= \infty$
iff  $\far\frac{e^3}{4\pi  \mEL c^2}  \geq \frac32 {\sqrt{G}\hbar}/{c}$,  i.e. then 
there is at least one non-square integrable solution to ${\widetilde{K}}_k^{\far}  g = \lambda g$ 
for each $\lambda \in \mathbb{C}$, or equivalently for a fixed $\lambda$, see \cite[Theorem~5.6]{WeiC}.  

 We start with the left neighborhood of $x = \infty$. 
 As $x \rightarrow \infty$, the operator ${\widetilde{K}}_k^{\far}$ approaches
\begin{align}
K_* :=\left[ \begin{array}{cc} - \frac{Z \alpha_s}{r_*} & - \frac{d}{dx}  \\ 
\frac{d}{dx}& - \frac{Z \alpha_s}{r_*}
	  \end{array} \right],
\end{align}
where again $r_*=r_-$ in the subextremal case and $r_*=r_0$ in the extremal case.
 Clearly,
$g_\pm= ( e^{\pm i \frac{Z \alpha_s}{r_*} x} , \mp i e^{\pm i \frac{Z \alpha_s}{r_*} x})^T$ are solutions to $K_*  g = 0$, and 
$g_\pm$ is not square integrable at $\infty$. 
 Hence, the LPC is satisfied  in the left neighborhood of $x= \infty$.

 Next, we address the right neighborhood of $x = 0$ ($r=0$), and consider the solutions to 
\begin{align} \label{0eigen}
&\Big[ f(r) -  \frac{Z \alpha_S}{r} \Big] g_1 + 
\Big[ \frac{k f(r)}{r} -   Z \alpha_S^2 \far \frac{f(r)}{4 \pi r^2}  - f^2(r) \frac{d}{dr} \Big] g_2=0 , \\
& \Big[ - f(r) -  \frac{Z \alpha_S}{r} \Big] g_2 + 
\Big[ \frac{k f(r)}{r} -   Z \alpha_S^2 \far \frac{f(r)}{4 \pi r^2}  + f^2(r) \frac{d}{dr} \Big] g_1=0.
\label{00eigen}
\end{align} 
 Recall that $g = (g_1, g_2)^T$ is square integrable in the right neighborhood of $r = 0$ with the inner product associated with
\eqref{fnormBH} iff for each $ 0 <R < r_-^{}$,
\begin{align} \label{sqint}
\int_{0}^{R} \frac{1}{f^2(r)}  \Big( |g_1(r)|^2 + |g_2(r)|^2 \Big) dr < \infty .
\end{align}
 Therefore, we aim to find solutions to \eqref{0eigen}, \eqref{00eigen} such that \eqref{sqint} does not hold.
 Note that as $r \rightarrow 0$, $f(r) \sim q r^{- 1}$, where $q= (r_- r_+)^{\frac{1}{2}}$ in the subextremal 
case and $q = r_0 $ in the extremal case, when $r_-=r_+ \;(=r_0)$.
 
 Hence, around zero equations \eqref{0eigen}, \eqref{00eigen} become
\begin{align}
& g_2^{\prime} +  \frac{Z \alpha_S^2\far}{4 \pi q r} g_2 - \frac{k}{q} g_2 = O(r)g_1,   \\
& g_1^{\prime} -  \frac{Z \alpha_S^2\far}{4 \pi q r } g_1+ \frac{k}{q}g_1 = O(r)g_2.
\end{align}
 The above equations imply that in a right neighborhood of $r = 0$ we 
can find solutions 
$g_1 \sim r^\frac{Z \alpha_S^2\far}{4 \pi q }$ and $ g_2 \sim r^{ - \frac{Z \alpha_S^2\far}{4 \pi q }}$.
  Note that  \eqref{sqint} implies that local square integrability holds for $g_1$ and $g_2$ iff
\be 
\int_0^{R} r^{ \pm \frac{Z \alpha_S^2\far}{2 \pi q }+ 2 } dr < \infty .
\ee
 Recalling the definition of $r_{\pm}$ from \eqref{rpm}, and $r_0$ from \eqref{r0} we see that $q = \frac{ G^{1/2} m_e Z e}{\hbar c}$ 
in both the subextremal and extremal case, so that $Z$ cancels out in the power of $r$.
 Therefore,  in both the subextremal and extremal case the LPC is satisfied iff 
$-\frac{1}{2\pi}\frac{\alpha_S^2 \far \hbar c }{ G^{1/2} m_e  e} +2 \leq -1$, which translates into
$\far\frac{e^3}{4\pi  \mEL c^2}  \geq \frac32 {\sqrt{G}\hbar}/{c}$. 
\end{proof}

 Inserting numerical values for the physical and mathematical constants, we conclude that 
$\far\geq 1.3\cdot 10^{-18}$ implies essential self-adjointness.
 Therefore we arrive at 
\begin{cor}
 ${\widetilde{K}}_k^{\far}$ is essentially self-adjoint if the empirical value of the electron's anomalous magnetic 
moment is used, in which case ${\far}=1$ to three significant digits.
\end{cor}

 In the rest of this section we characterize spec ${\widetilde{K}}_k^{\far}$ when 
$\far\frac{e^3}{4\pi  \mEL c^2} > \frac32 {\sqrt{G}\hbar}/{c}$. 

\begin{thm}\label{thm:critANOMmagnMOM}
The essential spectrum $\sigma_{ess} ({\widetilde{K}}_k^{\far}) = \mathbb{R}$ for both the subextremal and the extremal case.
\end{thm} 

\begin{proof}
We define the operators  ${\widetilde{K}}_k^{\far}({[0,b]})$ and
 $ {\widetilde{K}}_k^{\far}({[b, \infty)})$ as the restriction of  ${\widetilde{K}}_k^{\far} $ to $L^2([0,b])$ and $ L^2([b,\infty])$ respectively. 
Then by Theorem~11.5 in \cite{WeiC}, we have 
\be 
\sigma_{ess} ({\widetilde{K}}_k^{\far}) = \sigma_{ess} ({\widetilde{K}}_k^{\far}({[0,b]})) \cup  \sigma_{ess} ( {\widetilde{K}}_k^{\far}({[b, \infty)})).
\ee
 Instead of \eqref{weyl} we now use the following Weyl sequence,
\be 
f_{n, \lambda} (x) = \frac{1}{\sqrt{2 n }} e^{- \frac{(x-b)}{2n} + i w} 
\left[{\scriptsize{\begin{array}{c} 1 \\ -i 
	\end{array}}}\right];\quad n\in\Nset.
\ee
 Then $\xi(x) = - \int_b^{x} [kc(y) - \far d(y)] dy$ for $ b \leq x \leq b+1$ in \eqref{S}, and one can show that 
 $\sigma_{ess} ( {\widetilde{K}}^{[b, \infty)}_{\mu_a,k})=\mathbb{R}$ in a similar way as in the proof of Theorem~\ref{essK}. 

 On the other hand, the operator $ {\widetilde{K}}_k^{\far}({0, b]})$ can only have discrete spectrum. 
 To see that, we use Theorem~2 in \cite{HS}. 
In particular, since the limit point case holds, $ {\widetilde{K}}_k^{\far}({[0,b]})$ has discrete spectrum if also 
\be
 \int_0^{b} |kc(x) - \far d(x)| dx = \infty. 
\ee
 Notice that $d(x) \sim x^{-1}$ as $x\to 0$, which is not locally integrable around zero.
 See \eqref{Ktildeab} for the definitions of $c(x)$ and $d(x)$. 
\end{proof}

\begin{prop} \label{acK1} 
 In both the subextremal and the extremal case, 
$\mathbb{R}\setminus\{ -\frac{\alpha_s Z} {r_-}\} \subset \ \sigma_{ac}({\widetilde{K}}_{k}^\far)$ and $\sigma_{sc}({\widetilde{K}}_{k}^\far)=\emptyset$.
\end{prop}
\begin{proof}
We first note that the fact that $\sigma_{sc}({\widetilde{K}}_{k}^\far)=\emptyset$ follows from the claim on the absolutely continuous spectrum,
 see Corollary~\ref{cor:sc}.
 Therefore, we only prove that $\mathbb{R}\setminus\{ -\frac{\alpha_s Z} {r_-}\} \subset \ \sigma_{ac}({\widetilde{K}}_{k}^\far)$.
 Similarly to the proof of Proposition~\ref{acK}, we use Theorem~\ref{th:weid}.
 In particular, we now need to consider the operator $P_2^\far(x)$ defined in \eqref{def:p2}, with $\far >0$. 
 We already proved the limit property \eqref{lim0} for ${P}_2^\far$.
 And so, our proof of Proposition \ref{acK} in concert with Corollary \ref{cor:sc} also proves Proposition \ref{acK1}.
\end{proof} 

\begin{thm}
In the subextremal case ${\widetilde{K}}_k^{\far}$ has no eigenvalues. 
\end{thm} 

\begin{proof}
 The proof  follows similarly to the proof of Theorem~\ref{eigen}. 
 One needs to consider the operator in \eqref{eigop} with $kc(x)$ replaced by  $kc(x) - \far d(x)$.
 Note that $ d(x) \sim c(x)$ at $\infty$, and hence the integrability condition in Lemma~\ref{noeig} is satisfied. 
 This concludes that $g$ has to be identically zero, if $g \in L^2$ and $( {\widetilde{K}}_k^{\far}+ \frac{Z \alpha_s}{r_-}  I) g =0$. 
\end{proof} 

\begin{cor}\label{cor:ACanom} 
 In the subextremal case the continuous spectrum is purely absolutely continuous and given by the whole real line,
$\sigma_{ac}(K_{k}^{\far})= \Rset$.
\end{cor}

\section{Summary and outlook}\label{sec:summ}

 In this paper we have investigated the Dirac Hamiltonian of a test electron
with or without anomalous magnetic moment when the electron is assumed to reside in the static subregion of the interior of a
Reissner--Weyl--Nordstr\"om (RWN) black-hole spacetime. 
 Using the partial wave decomposition the Dirac Hamiltonian becomes a direct sum of so-called radial Dirac operators
${\widetilde{K}}_k^{\far}$,
 where $\far=0$ amounts to an electron without, and $\far\approx 1$ to an electron with empirical anomalous magnetic moment,
and so it suffices to study the radial Dirac operators. 
 We found that ${\widetilde{K}}_k^{\far}$ is essentially self-adjoint if 
$\far \geq \far_{\text{\tiny{crit}}}\approx 1.3\cdot 10^{-18}$,
and has infinitely many self-adjoint extensions when $\far=0$.
 (We expect infinitely many self-adjoint extensions for all $0\leq \far< \far_{\text{\tiny{crit}}}$, but we haven't verified this.)
 So when working with the empirical value of the electron's anomalous magnetic moment, i.e. $\far=1$ to several significant digits, 
this Dirac operator is well-defined and generates a unitary dynamics for the electron on the static subregion inside the RWN black hole. 

 We have characterized the spectrum of any self-adjoint extension in all subextremal cases where we showed they exist,
and  we found the essential spectrum is the whole real line, consisting of purely absolutely continuous spectrum.
 Since there is no gap in the continuum, no discrete Dirac spectrum exists in the subextremal black hole sector of RWN,
unlike the situation in the naked singularity sector. 
 Worse, the absolute continuity result for the spectrum means the complete absence of point spectrum for an electron in the 
static interior region of a subextremal RWN black-hole spacetime. 
 An analogous result was proved in \cite{CP} for an electron outside the event horizon of a subextremal RWN black hole.

 Since our study was inspired by the works \cite{CP}, \cite{Belgiorno}, and \cite{BMB} on the
general-relativistic spectrum of hydrogen and hydrogenic ions, we have chosen a narrative which portrays the charged
singularity of a RWN black-hole spacetime as a hyper-heavy point nucleus.
 Indeed, since hypothetical hyper-heavy nuclei by definition obey $GM\mEL >  Ze^2$, they also obey $GM^2 >  Z^2e^2$ 
(because $Z\mEL < Z\mPR \leq M =A(Z,N)\mPR$ with $A(Z,N)\approx Z+N$), 
and thus are associated with the black hole sector of the RWN  spacetime family.
  As we have shown in section \ref{sec:RWN}, this means that the number of neutrons $N\gg Z$, as for neutron stars. 
 Since neutron stars have been likened to gigantic nuclei bound by their overall gravity, and since
a non-rotating charged neutron star that accumulates more than a critical mass is expected to collapse and leave a RWN black-hole behind,
it is suggestive to think of the charged singularity of a RWN black hole as a hugely massive nucleus, 
albeit one that actually has degenerated into a point (distinct from the point particle approximation to an empirical
atomic nucleus).
 In terms of this narrative, combining the results of Cohen and Powers with ours we can conclude 
that there is no hyper-heavy hydrogenic ion point spectrum at all in the static regions of the subextremal RWN black-hole sector.

 It still remains to discuss the Cohen--Powers setup for the extremal sector, i.e. electron outside of the horizon --- but in this
paper we were only concerned with electrons in the static part of the interior region. 
 It also remains to settle the issue of the point spectrum in the extremal black hole case, when the electron spinor wave function is
supported inside the event horizon. 
 We have shown that the only possible eigenvalue is $-Z\alphaS/r_0$, where $r_0$ is the area radius at which the horizon is located, 
but it is not clear whether this value is an eigenvalue, and if so, whether it is simple, finitely degenerate, or even infinitely 
degenerate. 
 The extremal RWN black hole sector is not generic in the RWN spacetime family, but the open questions are technically challenging, and it
is curious to contemplate that this exceptional black hole setting is so far the only one which has not been ruled out of
permitting bound states.
 
 It also remains to be seen whether the presence of a horizon \emph{generically} causes absence of eigenvalues for the Dirac operator,
as conjectured in \cite{CP} for electrons with wave functions supported outside the event horizon, and which may now be conjectured to be
true also for electrons in the static interior of other subextremal black hole spacetimes. 
 To prove such a conjecture in all generality, if indeed true, is a challenging project. 
 Yet there are several feasible generalizations of our study which are worthy of pursuit,
and which could cement the conjecture further or, possibly, disprove it.

 One further direction of inquiry could be an investigation of the Dirac operator for a test electron in generalizations of the RWN 
black hole spacetime of a single point nucleus that obey other electrostatic vacuum laws. 
 The naked singularity sector of such spacetimes has been described in \cite{TZ}, and generalized in appendix B of \cite{KTT}.
 Such a study has the technical advantage that the spherical symmetry of the spacetime allows one to work with the partial wave decomposition
of the Dirac Hamiltonian, as done in the present paper. 

 For the naked singularity sector such a study has recently been carried out in \cite{Moulik} for singularities with zero bare mass, and
in \cite{KTT} for spacetimes with naked singularities of strictly negative bare mass, with some surprising results.
 The perhaps most surprising result of \cite{KTT} (to its authors at least) is that the Dirac operator for an electron in the naked 
singularity sector of the Hoffmann spacetime \cite{Hoffmann, TZ} (Born \cite{BornA} or Born--Infeld vacuum law) of a point nucleus 
is not essentially self-adjoint, with or without anomalous magnetic moment, unless the bare mass of the singularity vanishes \cite{Moulik}. 
 A vanishing bare mass is not typical, though, and so the upshot is that in the naked singularity sector of
the Hoffmann spacetime family the Dirac Hamiltonian of a test electron is typically not well-defined even if the anomalous magnetic
moment is taken into account.
 
 We suspect that the same conclusion will hold for the Dirac operator of a test electron in the static part of the 
interior region of a Hoffmann black hole spacetime.
 
 Another generalization of the present work is to study the Dirac equation for a test electron in the multi-black hole spacetime
family of Hartle and Hawking \cite{HH}, obtained by analytical completion of the asymptotically flat, static, Majumdar--Papapetrou metrics. 
 These are very special spacetimes,  but there are not many explicit representations of spacetimes with several black holes
in them. 
 Each black hole of the Hartle--Hawking family obeys the RWN extremal condition $GM^2 = Q^2$.  
 The simplest multi-black hole case is a two-black-holes spacetime, inevitably having axial symmetry. 
 Separation of variables should again be feasible, even though perhaps not as explicitly solvable as in the spherically symmetric case. 
 A discrete reflection symmetry is available in a three-black-holes spacetimes, offering some simplification, 
yet for three or more nuclei (black holes) functional and PDE analysis will have to be fielded to study 
the Dirac Hamiltonian, cf. \cite{EstebanA} and references therein.

\smallskip

\vfill
\hrule

\end{document}